\documentclass[11pt]{article}
\usepackage[a4paper,margin=26mm]{geometry}
\usepackage[T1]{fontenc}
\usepackage{lmodern,amsmath,amssymb,amsthm,mathtools,microtype}
\usepackage[colorlinks=true,linkcolor=blue,citecolor=blue,urlcolor=blue,hypertexnames=false]{hyperref}
\usepackage[nameinlink]{cleveref}
\usepackage{enumitem}
\usepackage{tikz,flafter}
\usetikzlibrary{arrows.meta}
\setlist{itemsep=3pt,topsep=5pt}
\newtheorem{theorem}{Theorem}[section]
\newtheorem{lemma}[theorem]{Lemma}
\newtheorem{fact}{Fact}
\newtheorem{claim}{Claim}

\newtheorem{proposition}[theorem]{Proposition}

\theoremstyle{remark}
\newcommand{\F}{\mathbb F}
\newcommand{\dom}{\operatorname{dom}}
\newcommand{\R}{\mathrm R}
\newcommand{\poly}{\mathrm{poly}}
\newcommand{\Q}{\mathrm Q}
\newcommand{\Shape}{\mathsf{Shape}}
\newcommand{\Verify}{\mathsf{V}}
\newcommand{\bits}{\{0,1\}}
\newcommand{\polylog}{\operatorname{polylog}}

\usepackage{mathpazo}
\title{Improved Separations between Quantum and Classical Communication Complexity of Total Functions}
\author{Fran\c{c}ois Le Gall
\\Graduate School of Mathematics\\Nagoya University\\
\texttt{legall@math.nagoya-u.ac.jp}}
\date{}
\begin{document}
\maketitle
\begin{abstract}
We refine Gavinsky's framework (arXiv:2608.18784) for exponential separations between quantum and randomized communication complexity of total functions and obtain larger separations: polylogarithmic quantum communication versus $\widetilde\Omega(\sqrt n)$ randomized communication with two quantum messages, and versus $\Omega(n^{1-\varepsilon})$ for every fixed $0<\varepsilon<1$ with more quantum messages. 

Our main technical contribution is a general conversion from quantum versus randomized communication separations for partial functions to separations for total functions. Our construction combines the communication lookup construction of Anshu, Belovs, Ben-David, G\"o\"os, Jain, Kothari, Lee, and Santha (FOCS 2016), the fully linear PCPs of Boneh, Boyle, Corrigan-Gibbs, Gilboa, and Ishai (CRYPTO 2019), and a linear seeded extractor with a short seed based on the construction of Raz, Reingold, and Vadhan (JCSS 2002). The key idea is to use the extractor to compress the lookup table into a single cell, thereby reducing the input length.
\end{abstract}\vspace{5mm}

\section{Introduction}
\paragraph{Background.}
Understanding the power of quantum communication is the central goal of quantum communication complexity.
Of particular interest is the case of total Boolean functions.
For a total function $f\colon\bits^n\times\bits^n\to\bits$,
write $\R(f)$ and $\Q(f)$ for its bounded-error complexities, and
$\Q^r(f)$ for quantum complexity with at most $r$ messages.
Seminal works~\cite{aa,bcw} gave a quadratic separation between $\R(f)$ and $\Q(f)$.
Using the cheat sheet technique from query complexity~\cite{abk}, Anshu, Belovs, Ben-David,
G\"o\"os, Jain, Kothari, Lee, and Santha~\cite{anshu} obtained
$\R(f)=\widetilde\Omega(\Q(f)^{2.5})$.\footnote{The notation $\widetilde O(\cdot)$, $\widetilde\Omega(\cdot)$ and $\widetilde\Theta(\cdot)$ suppresses factors polynomial in $\log n$.}
Tal's work~\cite{tal} implies an exponent of $8/3-o(1)$;
Bansal and Sinha~\cite{bansal-sinha} and Sherstov, Storozhenko, and
Wu~\cite{ssw} obtained
$\R(f)\geq\Q(f)^{3-o(1)}$.

Recently, Gavinsky~\cite{gavinsky} proved the first exponential separation
for a total function, using only two quantum messages:
\begin{theorem}[{\cite{gavinsky}}]\label{th:ga}
There is a total function $f\colon\bits^n\times\bits^n\to\bits$ with
$\Q^2(f)=\polylog(n)$ and $\R(f)=\widetilde\Omega(n^{1/6})$.
\end{theorem}
The construction from \cite{gavinsky} combines the communication lookup method of \cite{anshu} with efficient polynomial-certificate verification
based on fully linear PCPs~\cite{boneh}.

\paragraph{Statement of our results.}
We improve the exponent of Theorem \ref{th:ga} from $1/6$ to $1/2$:
\begin{theorem}\label{th:shape}\label{cor:shape}
There is a total function $f\colon\bits^n\times\bits^n\to\bits$ with
 $\Q^2(f)=\polylog(n)$ and $\R(f)=\widetilde\Omega(\sqrt n)$.
\end{theorem}
We also achieve every exponent below $1$ with a constant
number of quantum messages:\footnote{For $k=2$, Theorem~\ref{th:shape} gives the same randomized lower bound
with only two quantum messages.}
\begin{theorem}\label{th:forrelation}\label{cor:forrelation-total}
For every constant integer $k\geq2$, there is a total function $f\colon\bits^n\times\bits^n\to\bits$ with $\Q^{\,2\lceil k/2\rceil+1}(f)=\polylog(n)$ and
$\R(f)=\widetilde\Omega(n^{1-1/k})$.
\end{theorem}

Our main technical contribution is the following general conversion from quantum versus randomized communication
separations for partial functions to separations for total functions:

\begin{theorem}\label{thm:main}
Let $F:\bits^n\times\bits^n\to\{0,1,*\}$ be a partial function. Suppose its domain indicator
\[
 D_F(x,y)=\mathbf1[(x,y)\in\dom(F)]
\]
has a Boolean circuit of size at most $s$, where $s$ is a positive integer with $s\leq\poly(n)$, and $F$ has an $r$-message quantum protocol of cost $q\ge 1$, where $r\geq1$. There is a total function $F_{\mathrm{tot}}:\bits^N\times\bits^N\to\bits$ with
\[
 N=\widetilde O(n+s),\qquad
 \R(F_{\mathrm{tot}})=\widetilde\Omega(\R(F)),\qquad
 \Q^{r+1}(F_{\mathrm{tot}})=\widetilde O(q).
\]
\end{theorem}
This reduces the input-length overhead of the construction used in~\cite{gavinsky}, yielding larger separations. In particular, if $s=\widetilde O(n)$, then $N=\widetilde O(n)$.

\paragraph{Overview of our techniques.}
The starting point for our proof of Theorem~\ref{thm:main}
is the lookup construction of \cite{anshu},
as used in Gavinsky's framework~\cite{gavinsky}.
Informally, this construction takes~$c$ instances of $F$ and uses their~$c$ output bits as an address $L$.
Each party also receives a table with one cell for each of the $2^c$ possible
addresses; the selected cell supplies a candidate certificate that all $c$
input pairs satisfy the promise of $F$. Applying Gavinsky's construction
with $c\approx\log n$ gives $O(n)$ cells, each with a
certificate of length $\widetilde O(s)$, and hence input length
$\widetilde O(ns)$ per party.

Our main insight is to use a 
linear seeded extractor with a short seed, based on the construction of Raz, Reingold, and Vadhan \cite{rrv}, to replace the
lookup table by a single ``compressed'' cell.
Our construction uses
$c=\polylog(n)$ copies of $F$ and represents the $2^c$ candidate
certificates through this cell, whose length is $\widetilde O(n+s)$.
The cell is represented as $U+V$, where Alice holds $U$, Bob holds $V$,
and $+$ denotes bitwise XOR.
The candidate certificate at address $L$ is $E_L(U+V)$, where $E_L$ is
the linear map obtained by fixing the extractor seed to $L$.
Equation~\eqref{lb:lookup-definition} in Section~\ref{sec:lookup-setup} defines $F_{\mathrm{tot}}$ to be $1$
exactly when this candidate is a valid certificate, and $0$ otherwise;
if any input pair violates the promise, the output is also $0$.

Our construction thus removes the table-size factor $2^c$ from the input length, but still specifies $2^c$ candidate certificates $E_L(U+V)$, one for each possible seed
$L$. However, these certificates are correlated, so the lookup lower bound from \cite{anshu} 
does not apply directly. As we show in Section \ref{lb:section}, the extractor provides the near-uniformity needed
to adapt the proof of~\cite[Theorem~6]{anshu} and retain the randomized
lower bound up to polylogarithmic factors.

Applying Theorem~\ref{thm:main} to the function $\Shape_m$ from \cite{Gavinsky2} with the
$O(m^2)$-size promise-checking circuit used in~\cite{gavinsky} already gives
exponent $1/4$. Proposition~\ref{prop:shape-circuit} reduces this circuit size
to $\widetilde O(m)$, yielding exponent $1/2$ and proving
Theorem~\ref{th:shape}.

Theorem~\ref{th:forrelation} follows by applying
Theorem~\ref{thm:main} to Bansal--Sinha's $k$-Forrelation
function~\cite{bansal-sinha} composed with an inner product gadget.
The lifting theorem of~\cite{cfkmp} gives the randomized
communication lower bound, while
Proposition~\ref{prop:forrelation-domain} shows that the promise
of the composed function can be checked by a circuit of
size $\widetilde O(n)$.

\section{Preliminaries}
\paragraph{Definitions and conventions.}
Let $\F_2$ denote the binary field.
Write $\mathsf U_a$ for the uniform distribution on $\F_2^a$; tensor products
of such distributions indicate independent variables.  
A superscript denotes conditioning;
for example, $Z^{\mathbf x}$ is the distribution of $Z$ conditioned on
$\mathbf X=\mathbf x$. 
All logarithms are binary.
The min-entropy of $Z$ is $H_\infty(Z)=-\log\max_z\Pr[Z=z]$.

For distributions
$P,Q$ on the same finite space, define total variation and squared Hellinger
distance by
\[
 \Delta(P,Q)=\frac12\sum_z|P(z)-Q(z)|,
 \qquad
 h^2(P,Q)=1-\sum_z\sqrt{P(z)Q(z)}.
\]
We use
\begin{equation}\label{lb:distances}
 h^2(P,Q)\leq\Delta(P,Q)\leq\sqrt2\,h(P,Q).
\end{equation}
We write $I(A:B)$ for the mutual information of random variables $A,B$
and $I(A:B\mid C)$ for their conditional mutual information given $C$:
\[
 I(A:B)=H(A)-H(A\mid B),\qquad
 I(A:B\mid C)=H(A\mid C)-H(A\mid B,C),
\]
where $H$ denotes Shannon entropy. Entropy and mutual information are
measured in bits.

\paragraph{Two prediction lemmas.}
We will use the following two easy lemmas. 
\begin{lemma}[Prediction and total variation]\label{lem:prediction-tv}
Let $B$ be a binary random variable and let $Z$ be a random variable.
If $\Delta(ZB,Z\otimes\mathsf U_1)>\delta$ for some $\delta>0$, then
$B$ can be predicted from $Z$ with error less than $1/2-\delta$.
\end{lemma}
\begin{proof}
For each $z$ of positive probability, set $t_z=\Pr[B=1\mid Z=z]$.
Define $\phi(z)=1$ if $t_z\geq1/2$, and $\phi(z)=0$ otherwise;
define $\phi$ arbitrarily on values of zero probability.
Conditioned on $Z=z$, the prediction $\phi(z)$ succeeds with probability
\[
 \max\{t_z,1-t_z\}=\frac12+\left|t_z-\frac12\right|.
\]
Moreover, by the definition of total variation,
\[
 \Delta(ZB,Z\otimes\mathsf U_1)
 =\frac12\sum_z\Pr[Z=z]
   \left(\left|t_z-\frac12\right|
         +\left|(1-t_z)-\frac12\right|\right)
 =\sum_z\Pr[Z=z]\left|t_z-\frac12\right|.
\]
Then
\[
\begin{aligned}
 \Pr[\phi(Z)=B]
 &=\sum_z\Pr[Z=z]\Pr[\phi(Z)=B\mid Z=z]\\
 &=\sum_z\Pr[Z=z]\left(\frac12+\left|t_z-\frac12\right|\right)\\
 &=\frac12+\sum_z\Pr[Z=z]\left|t_z-\frac12\right|\\
 &=\frac12+\Delta(ZB,Z\otimes\mathsf U_1)\\
 &> \frac{1}{2}+\delta,
\end{aligned}
\]
which proves the claim.
\end{proof}

\begin{lemma}[Prediction and mutual information]\label{lem:predicting-a-bit}
Let $B$ be a binary random variable and let $X,T$ be random variables.
If $I(B:T\mid X)>\delta$ for some $\delta>0$, then $B$ can be predicted
from $(X,T)$ with error less than $1/2-\delta/2$.
\end{lemma}
\begin{proof}
We first show that, for any observation $Z$, there is a deterministic
prediction rule $\phi$ such that
\[
 \Pr[\phi(Z)\ne B]\leq\frac12 H(B\mid Z).
\]
The binary entropy function $h_2\colon[0,1]\to[0,1]$ satisfies the following:
\begin{equation}\label{prelim:binary-entropy-bound}
 h_2(t)\geq2\min\{t,1-t\}\qquad(0\leq t\leq1).
\end{equation}

For each observed value $z$ of positive probability, set
$t_z=\Pr[B=1\mid Z=z]$ and define $\phi(z)$ to be the more likely bit, as in 
Lemma \ref{lem:prediction-tv}.
The error conditioned on $Z=z$ is $\min\{t_z,1-t_z\}$.
By~\eqref{prelim:binary-entropy-bound},
\[
\begin{aligned}
 \Pr[\phi(Z)\ne B]
 &=\sum_z\Pr[Z=z]\min\{t_z,1-t_z\}\\
 &\leq\frac12\sum_z\Pr[Z=z]h_2(t_z)
 =\frac12 H(B\mid Z).
\end{aligned}
\]
Now take $Z=(X,T)$ and use
\[
 H(B\mid X,T)
 =H(B\mid X)-I(B:T\mid X)<1-\delta,
\]
since $H(B\mid X)\leq1$.
\end{proof}

\paragraph{Short-seed linear extractor with surjective maps.}

For integers $b,c,d\ge 1$, consider a map
\[
 E\colon\F_2^{d}\times\F_2^c\to\F_2^b\,.
\]
The second argument of $E$ represents the seed.
For any $j\in\F_2^c$, we define the fixed-seed map 
\[
E_j\colon\F_2^{d}\to\F_2^b,\qquad E_j(u)=E(u,j)\,.
\]
The map $E$ is a strong $(k,\varepsilon)$-seeded extractor if, for every
source $Z$ on $\F_2^{d}$ with min-entropy at least~$k$, choosing a seed
$J$ uniformly from $\F_2^c$ and independently of $Z$ gives
\begin{equation}\label{lb:def-extractor}
 \Delta\bigl((J,E_J(Z)),\mathsf U_c\otimes\mathsf U_b\bigr)
 \leq\varepsilon.
\end{equation}

The following specialization of the Raz--Reingold--Vadhan
construction~\cite{rrv} includes a simple modification that makes every fixed-seed map surjective.

\begin{lemma}\label{lem:extractor}
Fix a constant $0<\varepsilon\leq1/2$. There is a constant $K_E>0$
such that, for every sufficiently large integer $b$, the choice
$c=\lceil K_E\log^2 b\rceil$ admits a strong
$(2b,\varepsilon)$-seeded extractor
\[
 E:\F_2^{4b}\times\F_2^c\longrightarrow\F_2^b
\]
whose fixed-seed maps $E_j:\F_2^{4b}\to\F_2^b$ are linear over $\F_2$ and
surjective. 
\end{lemma}
\begin{proof}
Set $\eta=\varepsilon/3$. Since $\eta$ is constant,
$\log(b/\eta)=O(\log b)$. The extractor construction of
Raz--Reingold--Vadhan gives an explicit strong $(2b,\eta)$-seeded extractor
$E':\F_2^{4b}\times\F_2^{r_E}\to\F_2^b$ with
$r_E=O(\log^2(b/\eta))=O(\log^2 b)$.
Thus there is a constant $K_E>0$, depending only on $\varepsilon$,
such that $r_E\leq K_E\log^2 b$ for all sufficiently large $b$.
The  single-stage strong-extractor analysis~\cite[Lemma~15 and Proposition~10 (or Lemma 9 and Proposition 24 in the ECCC version)]{rrv} requires source min-entropy
$\tfrac32 b+O(\log^2 b)$, which is at most $2b$ for sufficiently
large $b$.
The underlying code is Reed--Solomon concatenated with Hadamard, which is linear; each fixed seed selects codeword coordinates, so every
$E'_j$ is linear. 

We next modify the extractor so that every fixed-seed map is surjective.
Let $W$ be uniform on $\F_2^{4b}$. Since $W$ has min-entropy $4b\geq2b$, the strong-extractor guarantee gives
\[
 \mathbb E_{j\leftarrow\mathsf U_{r_E}}
 \bigl[\Delta(E'_j(W),\mathsf U_b)\bigr]\leq\eta.
\]
If $E'_j$ is not surjective, its image is a proper linear subspace
of $\F_2^b$, so it contains at most half of the possible outputs.
Thus $\Delta(E'_j(W),\mathsf U_b)\geq1/2$ for each such seed.
The displayed inequality therefore implies that at most a $2\eta$
fraction of seeds give a nonsurjective map.

For each of these seeds, replace the map by the projection
$(w_1,\ldots,w_{4b})\mapsto(w_1,\ldots,w_b)$.
All fixed-seed maps are now linear and surjective.
For any source independent of the uniform seed, the old and new outputs
can differ only when one of the replaced seeds is chosen, which happens
with probability at most $2\eta$. Hence their joint seed-output
distributions have total variation distance at most $2\eta$.
By the triangle inequality, the modified extractor has error at most
$\eta+2\eta=\varepsilon$ for every source of min-entropy at least $2b$.

Finally, we need a seed of length $c=\lceil K_E\log^2 b\rceil$,
whereas the construction uses only $r_E\leq c$ bits.
Define $E$ using the first $r_E$ bits of the seed and ignoring the rest.
For a uniform $c$-bit seed, these first bits form a uniform $r_E$-bit
seed. The remaining bits are independent of both these bits and the
output, so including them does not change the statistical distance in
\eqref{lb:def-extractor}. Thus the error remains at most $\varepsilon$,
and every fixed-seed map remains linear and surjective.
\end{proof}

\paragraph{Communication complexity.}
Randomized communication complexity $\R(f)$ allows public randomness and arbitrarily many messages.
The quantum complexities $\Q(f)$ and $\Q^r(f)$ allow no prior entanglement.
In an $r$-message quantum protocol, the last receiver produces the output.
All costs are worst-case, measured in bits or qubits, respectively, with
error at most $1/3$ on each promised input.
For a randomized protocol $\Pi$, its communication cost
$\operatorname{CC}(\Pi)$ is the maximum number of bits exchanged over all
inputs and all outcomes of its randomness; shared random bits are not
counted as communication. If $\Pi$ computes a function $f$ and $\mu$ is
supported on $\dom(f)$, its distributional error is
\[
 \operatorname{err}_\mu(\Pi)
 =\mathbb E_{(x,y)\sim\mu}
   \Pr[\Pi(x,y)\ne f(x,y)],
\]
where $\Pi(x,y)$ denotes the output and the inner probability is over all
protocol randomness. The worst-case error of $\Pi$ is
\[
 \operatorname{err}(\Pi)
 =\max_{(x,y)\in\dom(f)}\Pr[\Pi(x,y)\ne f(x,y)].
\]

We will use the following property.
\begin{fact}[{\cite[Lemma~6.4]{bjks04}}]\label{fact:transcript-pythagorean}
For a randomized communication protocol, let $\mathsf T_{x,y}$ denote
the distribution of its communicated messages and public coins on the
fixed input pair $(x,y)$. Then, for any $x,x',y,y'$,
\begin{equation}\label{prelim:transcript-pythagorean}
 h^2(\mathsf T_{x,y},\mathsf T_{x',y})
  +h^2(\mathsf T_{x,y'},\mathsf T_{x',y'})
\leq2h^2(\mathsf T_{x,y},\mathsf T_{x',y'}).
\end{equation}
\end{fact}
In particular, if the transcript distributions on $(x,y)$ and $(x',y')$
are close, then they are also close on $(x,y)$ and $(x',y)$, and on $(x,y')$ and $(x',y')$.

\paragraph{Circuit-to-certificate construction.}
Here, Boolean circuits use fan-in-two AND and OR gates, and NOT gates.
We use the following consequence of the proof of
\cite[Lemma~5.1]{gavinsky}, based on fully linear PCPs~\cite{boneh}.
\begin{proposition}[{\cite[Lemma~5.1]{gavinsky}}]\label{prop:certificate}
Let $D(z)$ be a Boolean predicate computed by a circuit of size $S\geq2$. There is a predicate $\Verify_0(z;p)$, where the certificate $p\in\{0,1\}^{b_0}$ has positive length $b_0=O(S\log S)$.
The following properties hold, where a certificate $p$ is called valid
for $z$ if $\Verify_0(z;p)=1$, and invalid otherwise:
\begin{enumerate}
\item if $D(z)=1$, exactly one certificate is valid;\footnote{While not explicitly mentioned in \cite{gavinsky}, 
uniqueness follows from the construction: Validity forces
all claimed gate outputs to equal the actual outputs, by induction over
the gates. These values uniquely determine the two interpolating
polynomials, and validity requires the certificate polynomial to equal
their product. Its coefficient encoding is therefore unique.}
\item if $D(z)=0$, no certificate is valid;
\item given an additive sharing of the input $z$ and certificate $p$ between two parties, there is a randomized test that accepts with probability $1$ when the certificate
is valid for the input, and with probability at most $1/64$ otherwise.
The test uses one message of $O(\log S)$ bits, sent in either direction,
containing a random challenge independent of the input and certificate
and four affine-answer shares. The receiver decides acceptance.
\end{enumerate}
\end{proposition}

\section{From Partial to Total Functions: Proof of Theorem \ref{thm:main}}
In this section we prove Theorem \ref{thm:main}. Section~\ref{sec:lookup-setup} defines $F_{\mathrm{tot}}$.
Section~\ref{sec:lb} proves Theorem~\ref{thm:main} using the classical
lower bound of Theorem~\ref{lb:main}.
The proof of Theorem~\ref{lb:main} appears in Section~\ref{lb:section}.

\subsection{Definition of the total function}\label{sec:lookup-setup}

Since Theorem~\ref{thm:main} holds trivially when $\dom(F)=\varnothing$, we assume below that $\dom(F)\ne\varnothing$.

Let $F:\bits^n\times\bits^n\to\{0,1,*\}$ be a partial Boolean
function, and let $s$ bound the size of a circuit
computing its domain indicator. 
Fix $\varepsilon=10^{-12}$ and the
corresponding constant $K_E$ from Lemma~\ref{lem:extractor}. First set
\[
 t=n+s+2,\qquad b=\left\lceil A t\log^3 t\right\rceil,
\]
where $A$ is a sufficiently large absolute constant. Then set
\[
 c=\left\lceil K_E\log^2 b\right\rceil.
\]
Choose $A$ large enough that $b\geq\max\{n+1,40\}$ and all subsequent size requirements are satisfied.
Let
$E:\F_2^{4b}\times\F_2^c\to\F_2^b$ be the strong
$(2b,\varepsilon)$-seeded extractor from Lemma~\ref{lem:extractor}. 

For tuples $\mathbf x=(x_1,\ldots,x_c)$ and
$\mathbf y=(y_1,\ldots,y_c)$, let $D_c(\mathbf x,\mathbf y)$ indicate
that all $c$ pairs lie in $\dom(F)$, and, on that domain, define
\[
 L(\mathbf x,\mathbf y)
   =\bigl(F(x_1,y_1),\ldots,F(x_c,y_c)\bigr)\in\F_2^c.
\]

The tuple-domain predicate $D_c$ is computed by taking the AND of $c$
copies of the domain circuit for $F$, giving a circuit of size $O(cs)$.
Applying Proposition~\ref{prop:certificate} to this circuit gives a
certificate of length
\[
 b_0=O\bigl(cs\log(cs)\bigr).
\]
For sufficiently large $A$, this length satisfies $b_0\leq b$: increasing
$A$ enlarges $b$ linearly, while the certificate-length bound grows only
polylogarithmically in $A$.
Define $\Verify(\mathbf x,\mathbf y;p)$ on $b$-bit certificates by applying
the original predicate to the first $b_0$ bits and ignoring the remaining
bits. Since $b_0\ge 1$ and the original certificate is unique, if $D_c(\mathbf x,\mathbf y)=1$ then there exist both valid and invalid $b$-bit certificates; if $D_c(\mathbf x,\mathbf y)=0$, none is valid.


Give Alice $(\mathbf x,U)$ and Bob
$(\mathbf y,V)$, with $U,V\in\F_2^{4b}$, and define\footnote{We need to separate the cases $D_c(\mathbf x,\mathbf y)=1$ and $D_c(\mathbf x,\mathbf y)=0$ since $L(\mathbf x,\mathbf y)$ is defined only when $D_c(\mathbf x,\mathbf y)=1$.}
\begin{equation}\label{lb:lookup-definition}
 F_{\mathrm{tot}}(\mathbf x,U;\mathbf y,V)=
 \begin{cases}
 \Verify\bigl(\mathbf x,\mathbf y;E_{L(\mathbf x,\mathbf y)}(U+V)\bigr) 
       &\textrm{ if }D_c(\mathbf x,\mathbf y)=1,\\
 0    &\textrm{ if }D_c(\mathbf x,\mathbf y)=0.
 \end{cases}
\end{equation}
The sum $U+V$ is over $\F_2$. Note that there are exactly \[
N=cn+4b=\widetilde O(n+s)
\]
input bits per party.

\subsection{Lower bound statement and proof of Theorem~\ref{thm:main}}\label{sec:lb}
We use the assumptions of Theorem~\ref{thm:main} and the construction of
Section~\ref{sec:lookup-setup}. In particular, $b\geq n+1\geq\R(F)$;
the assumption $s\leq\poly(n)$ gives $b\leq\poly(n)$ and
$c=O(\log^2 n)$. Our main technical result is the following lower bound,
proved in Section~\ref{lb:section}.
\begin{theorem}\label{lb:main}
$\R(F_{\mathrm{tot}})=\widetilde\Omega(\R(F))$.
\end{theorem}

We now show the quantum upper bound.

\begin{theorem}\label{thm:quantum-upper}
$\Q^{r+1}(F_{\mathrm{tot}})=\widetilde O(q)$.
\end{theorem}
\begin{proof}
Amplify each base protocol to error at most $1/(64c)$ by $O(\log c)$ independent repetitions. Run all $c$ slots and all repetitions in parallel. This costs $O(cq\log c)$ qubits while preserving the number of messages. The final receiver obtains an address $J$. If $D_c(\mathbf x,\mathbf y)=1$, a union bound gives $\Pr[J\ne L(\mathbf x,\mathbf y)]\leq1/64$.

By linearity, the certificate at address $j$ is additively shared as
$E_j(U)+E_j(V)$. Only after measuring and obtaining $J$ does the final receiver
sample a random challenge for the test of Proposition~\ref{prop:certificate}, using randomness independent of the preceding protocol. It sends $J$, the challenge, and its four affine verification shares in one final message, costing $O(c+\log(cs))=O(c+\log s)$ bits, since the circuit for $D_c$ has size $O(cs)$. The other party computes its own certificate share at $J$, reconstructs the four values, and outputs the test's decision.

If $D_c(\mathbf x,\mathbf y)=1$, conditioning on $J=L(\mathbf x,\mathbf y)$ leaves an error of at most $1/64$, so total error is at most $2/64$. If $D_c(\mathbf x,\mathbf y)=0$, every certificate is invalid for every possible $J$, so soundness bounds acceptance by $1/64$. This proves correctness on every input of the total function. The overall cost is
\[
 O(cq\log c+c+\log s )=\widetilde O(q),
\]
as claimed.
\end{proof}
Combining these two bounds with $N=\widetilde O(n+s)$ proves Theorem~\ref{thm:main}.

\begin{samepage}
\section{Proof of the Compressed Lookup Lower Bound}\label{lb:section}

\theoremstyle{plain}
\newtheorem{lookupconversion}[theorem]{Theorem}
\theoremstyle{remark}
In this section we consider the construction of Section~\ref{sec:lookup-setup}. Theorem~\ref{lb:main} follows from the theorem below (recall that $b\leq\poly(n)$).

\begin{lookupconversion}\label{lb:conversion}
For every $1/3$-error randomized protocol $\Pi$ for $F_{\mathrm{tot}}$, there is a $1/3$-error
randomized protocol $\Pi'$ for $F$ such that
\[
 \operatorname{CC}(\Pi')
 \leq O\!\left(\log^4 b\,\cdot\operatorname{CC}(\Pi)\right).
\]
\end{lookupconversion}
\end{samepage}

This result is obtained by adapting the proof of~\cite[Theorem~6]{anshu}.
The main proof is given in Section \ref{lb:main-argument}.
Lemma~\ref{lb:input-conditioning}, used in the main argument, is proved
in Section~\ref{lb:reuse-claims}.

\subsection{Proof of Theorem~\ref{lb:conversion}}\label{lb:main-argument}
Set $\delta=1/(10^{22}c)$.
We also write $\Pi$ for the transcript of the protocol $\Pi$, consisting of the
communicated messages and public coins $\mathcal R$.
The private randomness and the inputs $\mathbf x,\mathbf y,U,V$ are not
included in the transcript.
If $\operatorname{CC}(\Pi)\geq\R(F)$, an optimal $1/3$-error protocol
for $F$ already has cost at most $\operatorname{CC}(\Pi)$.
Henceforth assume $\operatorname{CC}(\Pi)<\R(F)$.

\paragraph{Fixing a distribution $\boldsymbol{\mu}$.}
For each distribution $\mu$ supported on $\dom(F)$, we seek a protocol
$\Pi_\mu$ for $F$ with
\begin{equation}\label{lb:distribution-target}
 \operatorname{CC}(\Pi_\mu)\leq \operatorname{CC}(\Pi)+1,
 \qquad
 \operatorname{err}_\mu(\Pi_\mu)\leq\frac12-\frac\delta3.
\end{equation}
As we will show later, the minimax principle and error reduction will then give the required
protocol $\Pi'$ on all promised inputs. Fix the distribution $\mu$
for the rest of the distributional argument.

\paragraph{Construct a distribution for \texorpdfstring{$\boldsymbol{F_{\mathrm{tot}}}$}{Ftot}.}
Define the joint input random variable
\[
 T=(X_1,\ldots,X_c,U,Y_1,\ldots,Y_c,V),
\]
where the pairs $(X_i,Y_i)$ are independent samples from $\mu$, and
$U,V$ are independent uniform elements of $\F_2^{4b}$, independent
of all these pairs. Write
\[
 \mathbf X=(X_1,\ldots,X_c),\quad
 \mathbf Y=(Y_1,\ldots,Y_c),\quad
 L_i=F(X_i,Y_i),\quad L=(L_1,\ldots,L_c).
\]
For each $\ell\in\F_2^c$, set $U_\ell=E_\ell(U)$ and $V_\ell=E_\ell(V)$.
The random variables \(U_\ell\) and \(V_\ell\) are independent and uniformly distributed on \(\F_2^b\), since \(U,V\) are independent uniform sources and \(E_\ell\) is linear and surjective.
We use $\mathbf X_{-i}$ for $\mathbf X$ with coordinate $i$ deleted
and $\mathbf X_{<i}=(X_1,\ldots,X_{i-1})$, with the same convention
for $\mathbf Y$ and $L$.

\paragraph{Rule out easy distributions \texorpdfstring{$\boldsymbol{\mu}$}{mu}.}
An input distribution is easy if one party can predict the output
from its own input with sufficient advantage. Let $W$ be uniform on
$\{0,1\}^c$, and consider the conditions
\begin{equation}\label{lb:a2}
 \Delta(\mathbf XL,\mathbf X\otimes W)\leq c\delta/3,
 \qquad
 \Delta(\mathbf YL,\mathbf Y\otimes W)\leq c\delta/3.
\end{equation}
Suppose the first condition fails (a similar argument can be used when the second condition fails). Set
$d=\Delta(X_1L_1,X_1\otimes\mathsf U_1)$. Since the pairs
$(X_i,L_i)$ are independent and identically distributed, replacing
one factor at a time and applying the triangle inequality gives
\[
 \Delta(\mathbf XL,\mathbf X\otimes W)
 \leq\sum_{i=1}^c\Delta(X_iL_i,X_i\otimes\mathsf U_1)
 =cd.
\]
Thus $d>\delta/3$. By Lemma~\ref{lem:prediction-tv}, Alice can predict
$L_1$ from $X_1$ with success probability greater than $1/2+\delta/3$.
On input $(x,y)\sim\mu$, Alice applies this prediction rule to $x$
and sends her prediction to Bob. This one-bit protocol~$\Pi_\mu$
satisfies~\eqref{lb:distribution-target}. 

We henceforth work under
\eqref{lb:a2}.

\paragraph{Construct new protocols \texorpdfstring{$\boldsymbol{\Pi_i}$}{Pi i}.}
For each $i\in[c]$, define a protocol $\Pi_i$ on inputs in $\dom(F)$.
Given $(X_i,Y_i)$, the parties use public coins to sample
$(\mathbf X_{-i},\mathbf Y_{-i})$ from $\mu^{\otimes(c-1)}$ and to
sample the independent uniform sources $U,V$. They run $\Pi$ on
$(\mathbf X,U;\mathbf Y,V)$.
The transcript, including the public coins, is
\[
 \Pi_i=(\Pi,\mathbf X_{-i},U,\mathbf Y_{-i},V).
\]
We obviously have 
\[
\operatorname{CC}(\Pi_i)\leq \operatorname{CC}(\Pi).
\]

\paragraph{Rule out informative protocols \texorpdfstring{$\boldsymbol{\Pi_i}$}{Pi i}.}
We test whether any of the new protocols gives a party enough
information about the output $L_i$. The conditions are
\begin{equation}\label{lb:a3}
 I(L_i:\Pi_i\mid X_i)\leq\delta,
 \qquad I(L_i:\Pi_i\mid Y_i)\leq\delta
 \qquad(i\in[c]).
\end{equation}
Suppose $I(L_i:\Pi_i\mid X_i)>\delta$ for some $i$ (a similar argument can be used if the second inequality does not hold).
After running $\Pi_i$, Lemma~\ref{lem:predicting-a-bit} lets Alice
predict $L_i$ from $(X_i,\Pi_i)$ with error less than
$1/2-\delta/2<1/2-\delta/3$. Sending this prediction to Bob adds one bit to
$\Pi_i$, producing a protocol $\Pi_\mu$ satisfying
\eqref{lb:distribution-target}. 

We can therefore assume
\eqref{lb:a3}.

\paragraph{Obtain a contradiction.}
If none of the preceding cases provides the required protocol, then Assumptions
\eqref{lb:a2} and \eqref{lb:a3} both hold. We show that this is impossible. 

For each $(\mathbf x,\mathbf y)$ in the support of
$\mathbf X\mathbf Y$, write
\[
 Q_{\mathbf x,\mathbf y}
 =\Pi^{\mathbf x,\mathbf y}.
\]
For the same $(\mathbf x,\mathbf y)$, and $\hat u,\hat v\in\F_2^b$, write
\[
 Q^{\hat u,\hat v}_{\mathbf x,\mathbf y}
 =\Pi^{\mathbf x,\mathbf y,\hat u,\hat v},
\]
where the additional superscripts $\hat u,\hat v$ mean conditioning on
$U_\ell=\hat u,V_\ell=\hat v$, where $\ell=L(\mathbf x,\mathbf y)$. We only need one such $(\mathbf x,\mathbf y)$ for which fixing
these values usually changes the transcript distribution very little. The following lemma, proved in Section~\ref{lb:reuse-claims}, guarantees that such a pair exists.

\begin{lemma}\label{lb:input-conditioning}
Under assumptions~\eqref{lb:a2} and~\eqref{lb:a3}, there exists
$(\mathbf x,\mathbf y)$ in the support of $\mathbf X\mathbf Y$ such that
\begin{equation}\label{lb:fiber-tail}
 \Pr_{\hat u,\hat v\leftarrow\mathsf U_b}\!\left[
  \Delta\bigl(Q^{\hat u,\hat v}_{\mathbf x,\mathbf y},
              Q_{\mathbf x,\mathbf y}\bigr)>\tau
 \right]<10^{-6}\,,
\end{equation}
where $ \tau=6\cdot10^6\sqrt{c\delta}=6\cdot10^{-5}$.
Here $\hat u,\hat v$ are sampled independently.
\end{lemma}

Fix $(\mathbf x,\mathbf y)$ supplied by
Lemma~\ref{lb:input-conditioning}, and set again $\ell=L(\mathbf x,\mathbf y)$. We assume they are now known to both parties. 
Since $\mu$ is supported on $\dom(F)$,
we have $D_c(\mathbf x,\mathbf y)=1$.

Consider an auxiliary protocol in which Alice receives
$\hat u\in\F_2^b$ and Bob receives $\hat v\in\F_2^b$.
Alice privately samples $u$ uniformly
from $E_\ell^{-1}(\hat u)$, and Bob independently samples $v$
uniformly from $E_\ell^{-1}(\hat v)$. They run $\Pi$ on
$(\mathbf x,u)$ and $(\mathbf y,v)$.
Alice's sample $u$ is distributed exactly as $U$
conditioned on $U_\ell=\hat u$, and similarly for Bob.
The two samples are independent, so running~$\Pi$ on them
gives transcript distribution $Q^{\hat u,\hat v}_{\mathbf x,\mathbf y}$.
Surjectivity ensures that the sampling sets are
nonempty, and linearity gives
$E_\ell(u+v)=\hat u+\hat v$.
This protocol therefore computes $g(\hat u+\hat v)$ with error at most $1/3$, where
$g(p)=\Verify(\mathbf x,\mathbf y;p)$.
Assume, by exchanging the parties if necessary, that Bob produces the output. 

We will use the following claim.
\begin{claim}\label{lb:opposite-inputs}
There exist $\hat u,\hat v,\hat u',\hat v'\in\F_2^b$ such that $g(\hat u+\hat v)=g(\hat u'+\hat v')\ne g(\hat u'+\hat v)$ and
\[
 \Delta(Q^{\hat u,\hat v}_{\mathbf x,\mathbf y},
        Q^{\hat u',\hat v'}_{\mathbf x,\mathbf y})\leq2\tau.
\]
\end{claim}
\begin{proof}
Choose $a\in\{0,1\}$ such that $|g^{-1}(a)|\geq2^{b-1}$.
For each $w\in\F_2^b$, define
\[
 \beta(w)=\Pr_{\hat u\leftarrow\mathsf U_b}
 \bigl[\Delta(Q^{\hat u,\hat u+w}_{\mathbf x,\mathbf y},Q_{\mathbf x,\mathbf y})>\tau\bigr].
\]
Equation~\eqref{lb:fiber-tail} can be rewritten as
\[
 \frac{1}{2^b}\sum_{w\in\F_2^b}\beta(w)<10^{-6}.
\]
Since $\beta(w)\geq0$ and $|g^{-1}(a)|\geq2^{b-1}$,
\[
 \frac{1}{|g^{-1}(a)|}\sum_{w\in g^{-1}(a)}\beta(w)
 \leq
 \frac{2^b}{|g^{-1}(a)|}
 \left(\frac{1}{2^b}\sum_{w\in\F_2^b}\beta(w)\right)
 <2\cdot10^{-6}.
\]
Thus there exists $w\in g^{-1}(a)$ such that
\[
 \Pr_{\hat u\leftarrow\mathsf U_b}
 \bigl[\Delta(Q^{\hat u,\hat u+w}_{\mathbf x,\mathbf y},Q_{\mathbf x,\mathbf y})>\tau\bigr]<2\cdot10^{-6}.
\]

Since $D_c(\mathbf x,\mathbf y)=1$, there are both valid and invalid certificates for $(\mathbf x,\mathbf y)$, so $g$ takes both values.
Choose $w'$ with $g(w')=1-a$. Since translation preserves the uniform
distribution, a union bound gives a $\hat u$ for which both
\[
\begin{aligned}
 \Delta(Q^{\hat u,\hat u+w}_{\mathbf x,\mathbf y},Q_{\mathbf x,\mathbf y})&\leq\tau,\qquad
 \Delta(Q^{\hat u+w+w',\hat u+w'}_{\mathbf x,\mathbf y},Q_{\mathbf x,\mathbf y})&\leq\tau.
\end{aligned}
\]
Set $\hat u'=\hat u+w+w'$, $\hat v=\hat u+w$, and $\hat v'=\hat u+w'$.
The triangle inequality gives the distance bound. Moreover,
$g(\hat u+\hat v)=g(\hat u'+\hat v')=g(w)=a$ and $g(\hat u'+\hat v)=g(w')=1-a$,
which proves the claim.
\end{proof}

Fix $\hat u,\hat v,\hat u',\hat v'$ supplied by Claim~\ref{lb:opposite-inputs}.
Fact~\ref{fact:transcript-pythagorean}, together with~\eqref{lb:distances}, gives
\[
\begin{aligned}
 \Delta(Q^{\hat u,\hat v}_{\mathbf x,\mathbf y},Q^{\hat u',\hat v}_{\mathbf x,\mathbf y})
 &\leq\sqrt{2h^2(Q^{\hat u,\hat v}_{\mathbf x,\mathbf y},Q^{\hat u',\hat v}_{\mathbf x,\mathbf y})}
 \leq\sqrt{4h^2(Q^{\hat u,\hat v}_{\mathbf x,\mathbf y},Q^{\hat u',\hat v'}_{\mathbf x,\mathbf y})}
 \leq\sqrt{8\tau}<1/3.
\end{aligned}
\]
Bob has the same input $\hat v$ in both cases.
Conditioned on~$\hat v$ and the transcript, the distribution of his local
randomness (including the sampled source~$v$) is independent of Alice's
input.
Consequently, Bob's output is obtained from the transcript by the same
randomized map on both inputs. Bob’s probabilities of outputting \(1\) cannot differ by more
than $\Delta(Q^{\hat u,\hat v}_{\mathbf x,\mathbf y},Q^{\hat u',\hat v}_{\mathbf x,\mathbf y})<1/3$.
But Claim~\ref{lb:opposite-inputs} gives $g(\hat u+\hat v)\ne g(\hat u'+\hat v)$, so correctness
requires Bob's probabilities of outputting~$1$ on these two inputs to
differ by at least $1/3$. This leads to a contradiction.

\paragraph{Minimax argument.}
For every $\mu$ on $\dom(F)$ we have obtained a protocol $\Pi_\mu$
satisfying \eqref{lb:distribution-target}. The minimax
principle gives a protocol $\widetilde\Pi$ for $F$ with
\[
 \operatorname{CC}(\widetilde\Pi)\leq \operatorname{CC}(\Pi)+1,
 \qquad
 \operatorname{err}(\widetilde\Pi)\leq\frac12-\frac\delta3.
\]
Error reduction by independent repetition and
majority voting then gives a $1/3$-error protocol~$\Pi'$ with
\[
 \operatorname{CC}(\Pi')
 \leq O\bigl(\delta^{-2}(\operatorname{CC}(\Pi)+1)\bigr)=O(c^2 \operatorname{CC}(\Pi)).
\]
Together with $c=\Theta(\log^2 b)$,
this establishes the claimed bound on
$\operatorname{CC}(\Pi')$.


\subsection{Proof of Lemma \ref{lb:input-conditioning}}\label{lb:reuse-claims}
Set $\alpha=\sqrt{c\delta}$. For $\mathbf x$ in the support of $\mathbf X$
and a fixed seed $\ell$, write
\[
\begin{aligned}
 P_{\mathbf x,\ell}
 &=\mathcal L(\Pi,E_\ell(U)\mid\mathbf X=\mathbf x),\\
 \widehat P_{\mathbf x,\ell}
 &=\mathcal L(\Pi,E_\ell(U)\mid\mathbf X=\mathbf x,L=\ell),
\end{aligned}
\]
where $\mathcal L$ denotes distribution. 
We use $\widehat P_{\mathbf x,\ell}$ only when
$\Pr[\mathbf X=\mathbf x,L=\ell]>0$.

\paragraph{Applying the extractor with a uniform seed.}
Although $U$ is initially uniform, conditioning on the transcript
may change its distribution, since the messages can reveal
information about~$U$.
We first show that it usually retains enough
min-entropy to apply the extractor. Fix $\mathbf X=\mathbf x$ and the
public coins $\mathcal R=r$, and let $Z$ be the communicated transcript.
There are at most $2^{\operatorname{CC}(\Pi)}$ possible values of $Z$.
Call a transcript $z$ unlikely if
$\Pr[Z=z\mid\mathbf x,r]<2^{-2b}$.
We have
\[
 \sum_{\substack{\textrm{unlikely }z}}
 \Pr[Z=z\mid\mathbf x,r]
 \leq 2^{\operatorname{CC}(\Pi)}\cdot2^{-2b}
 =2^{\operatorname{CC}(\Pi)-2b}.
\]

For every $z$ that is not unlikely,
Bayes' rule gives
\begin{align*}
 \Pr[U=u\mid\mathbf x,r,z]
  &=\frac{\Pr[U=u\mid\mathbf x,r]\,
\Pr[Z=z\mid\mathbf x,r,U=u]}
        {\Pr[Z=z\mid\mathbf x,r]}\\
 &=\frac{2^{-4b}\Pr[Z=z\mid\mathbf x,r,U=u]}
        {\Pr[Z=z\mid\mathbf x,r]}\\
 &\leq2^{-2b},
\end{align*}
where the second equality uses that \(U\) is uniform and independent of \(\mathbf X,\mathcal R\).
Thus $U$ conditioned on $\mathbf x,r,z$ has min-entropy at least $2b$.

For every transcript $z$ that is not unlikely, the strong-extractor
guarantee therefore gives
\[
 \mathbb E_{\ell\leftarrow\mathsf U_c}
 \Delta\bigl(\mathcal L(E_\ell(U)\mid\mathbf x,r,z),
             \mathsf U_b\bigr)
 \leq\varepsilon,
\]
where $\ell$ is sampled independently of all inputs and protocol coins.
For unlikely transcripts, we use the bound $\Delta\leq1$, so their
contribution is at most their total probability.
Averaging over $z,r,\mathbf x$ gives
\begin{equation}\label{lb:leakage-bound}
 \mathbb E_{\mathbf x\leftarrow\mathbf X,\,\ell\leftarrow\mathsf U_c}
 \Delta(P_{\mathbf x,\ell},\Pi^{\mathbf x}\otimes\mathsf U_b)
 \leq e,
 \qquad e=\varepsilon+2^{\operatorname{CC}(\Pi)-2b}<2\cdot10^{-12}.
\end{equation}
The last inequality follows from
$\operatorname{CC}(\Pi)<\R(F)\leq n+1\leq b$ and $b\geq40$, which give
$2^{\operatorname{CC}(\Pi)-2b}\leq2^{-b}\leq2^{-40}<10^{-12}=\varepsilon$.

\paragraph{The actual address.}
Assumption~\eqref{lb:a2} lets us replace the independent uniform seed by
the address $L$ in the expectation above. This changes its expectation by at most $c\delta/3$:
\begin{equation}\label{lb:actual-address-extractor}
 \mathbb E_{(\mathbf x,\ell)\leftarrow\mathbf XL}
 \Delta(P_{\mathbf x,\ell},\Pi^{\mathbf x}\otimes\mathsf U_b)
 \leq e+c\delta/3.
\end{equation}
Fix $i\in[c]$. Recall that
$\Pi_i=(\Pi,\mathbf X_{-i},U,\mathbf Y_{-i},V)$.
The chain rule gives
\[
\begin{aligned}
 I(L_i:\Pi_i\mid X_i)
 &=I(L_i:\mathbf X_{-i}\mid X_i)
   +I(L_i:\Pi U\mathbf Y_{-i}V\mid\mathbf X)
 =I(L_i:\Pi U\mathbf Y_{-i}V\mid\mathbf X).
\end{aligned}
\]
The first term vanishes because $L_i=F(X_i,Y_i)$ depends only on the
$i$-th input pair, which is independent of all the other input pairs.

Given $\mathbf X$ and $\mathbf Y_{-i}$, we can compute
$L_{<i}=(F(X_1,Y_1),\ldots,F(X_{i-1},Y_{i-1}))$.
Discarding information cannot increase mutual information, so
\[
 I(L_i:\Pi U\mathbf Y_{-i}V\mid\mathbf X)
 \geq I(L_i:\Pi U L_{<i}\mid\mathbf X).
\]
Applying the chain rule once more gives
\[
\begin{aligned}
 I(L_i:\Pi U L_{<i}\mid\mathbf X)
 &=I(L_i:L_{<i}\mid\mathbf X)
   +I(L_i:\Pi U\mid\mathbf X,L_{<i})
 =I(L_i:\Pi U\mid\mathbf X,L_{<i}).
\end{aligned}
\]
Here the first term vanishes because, even after conditioning on
$\mathbf X$, the variables $L_1,\ldots,L_c$ are independent.
Combining these relations with assumption~\eqref{lb:a3}, we obtain
\[
 I(L_i:\Pi U\mid\mathbf X,L_{<i})
 \leq I(L_i:\Pi_i\mid X_i)\leq\delta.
\]
Finally, the chain rule and summation over $i$ yield
\[
 I(L:\Pi U\mid\mathbf X)
 =\sum_{i=1}^c I(L_i:\Pi U\mid\mathbf X,L_{<i})
 \leq c\delta.
\]
Pinsker's inequality and the Cauchy--Schwarz inequality therefore give
\begin{equation}\label{lb:address-conditioning-tv}
\begin{aligned}
 &\mathbb E_{(\mathbf x,\ell)\leftarrow\mathbf XL}
 \Delta\bigl((\Pi,U)^{\mathbf x,\ell},(\Pi,U)^{\mathbf x}\bigr)
 \leq\sqrt{I(L:\Pi U\mid\mathbf X)}
 \leq\sqrt{c\delta}=\alpha.
\end{aligned}
\end{equation}
Applying $(\pi,u)\mapsto(\pi,E_\ell(u))$ cannot increase total variation
 distance, so~\eqref{lb:address-conditioning-tv} implies
\begin{equation}\label{lb:extracted-conditioning-tv}
 \mathbb E_{(\mathbf x,\ell)\leftarrow\mathbf XL}
 \Delta\bigl(\widehat P_{\mathbf x,\ell},P_{\mathbf x,\ell}\bigr)
 \leq\alpha.
\end{equation}
Discarding $U$ also cannot increase total variation distance, so
\eqref{lb:address-conditioning-tv} also gives
\begin{equation}\label{lb:transcript-conditioning-tv}
 \mathbb E_{(\mathbf x,\ell)\leftarrow\mathbf XL}
 \Delta\bigl(\Pi^{\mathbf x,\ell},\Pi^{\mathbf x}\bigr)
 \leq\alpha.
\end{equation}

For each $(\mathbf x,\ell)$ in the support of $\mathbf XL$, the triangle
inequality gives
\[
\begin{aligned}
 \Delta\bigl(\widehat P_{\mathbf x,\ell},
               \Pi^{\mathbf x,\ell}\otimes\mathsf U_b\bigr)
 \leq
 \Delta\bigl(\widehat P_{\mathbf x,\ell},P_{\mathbf x,\ell}\bigr)
 +
 \Delta\bigl(P_{\mathbf x,\ell},\Pi^{\mathbf x}\otimes\mathsf U_b\bigr)
 +
 \Delta\bigl(\Pi^{\mathbf x}\otimes\mathsf U_b,
               \Pi^{\mathbf x,\ell}\otimes\mathsf U_b\bigr).
\end{aligned}
\]
Averaging over $(\mathbf x,\ell)\leftarrow\mathbf XL$ and applying
\eqref{lb:extracted-conditioning-tv}, \eqref{lb:actual-address-extractor},
and~\eqref{lb:transcript-conditioning-tv} bounds the three terms by
$\alpha$, $e+c\delta/3$, and $\alpha$, respectively.
Expanding the left-hand side as an average over the transcript gives
\begin{equation}\label{lb:alice-average}
 \mathbb E_{(\mathbf x,\ell,\pi)\leftarrow\mathbf XL\Pi}
 \Delta\bigl(U_\ell^{\mathbf x,\ell,\pi},\mathsf U_b\bigr)
 =
 \mathbb E_{(\mathbf x,\ell)\leftarrow\mathbf XL}
 \Delta\bigl(\widehat P_{\mathbf x,\ell},
            \Pi^{\mathbf x,\ell}\otimes\mathsf U_b\bigr)
 \leq2\alpha+e+c\delta/3.
\end{equation}

\paragraph{Fixing $\boldsymbol{(\mathbf x,\mathbf y)}$.}
We now include $\mathbf y$ in the conditioning in~\eqref{lb:alice-average}.
Given $\mathbf X=\mathbf x,L=\ell$, the variable~$U$ is independent of
$(\mathbf Y,V)$. This independence persists after conditioning on the
transcript: for fixed public coins and messages, the probability of
those messages factors into a function of Alice's input and a function
of Bob's input. Thus learning $\mathbf y$ in addition to
$\mathbf x,\ell,\pi$ does not change the distribution of $U_\ell$.
We may therefore insert this conditioning and average over
$\mathbf y$ to obtain
\[
 \mathbb E_{(\mathbf x,\mathbf y,\ell,\pi)\leftarrow\mathbf X\mathbf YL\Pi}
 \Delta\bigl(U_\ell^{\mathbf x,\mathbf y,\ell,\pi},\mathsf U_b\bigr)
 \leq2\alpha+e+c\delta/3.
\]
The same argument with Alice and Bob exchanged gives
the same bound for $V_\ell$.

For fixed $\mathbf x,\mathbf y,\ell,\pi$, the variables $U_\ell,V_\ell$
are independent (but need not be uniform), by the same factorization of transcript probabilities.
We can therefore apply the elementary product bound
\[
 \Delta(A\otimes B,\mathsf U_b\otimes\mathsf U_b)
 \leq\Delta(A,\mathsf U_b)+\Delta(B,\mathsf U_b)
\]
to their conditional distributions and average. The resulting bound is
\[
 \mathbb E_{(\mathbf x,\mathbf y,\ell,\pi)\leftarrow\mathbf X\mathbf YL\Pi}
 \Delta\bigl((U_\ell,V_\ell)^{\mathbf x,\mathbf y,\ell,\pi},
             \mathsf U_b\otimes\mathsf U_b\bigr)
 \leq4\alpha+2e+2c\delta/3<5\cdot10^{-11},
\]
where we used $\alpha=10^{-11}$, $c\delta=10^{-22}$, and
$e<2\cdot10^{-12}$.
At least one pair in the support must therefore satisfy
\[
 \mathbb E_{\pi\leftarrow\Pi^{\mathbf x,\mathbf y,\ell}}
 \Delta\bigl((U_\ell,V_\ell)^{\mathbf x,\mathbf y,\ell,\pi},
             \mathsf U_b\otimes\mathsf U_b\bigr)<5\cdot10^{-11},
\]
where $\ell=L(\mathbf x,\mathbf y)$. Fix this pair.

\paragraph{Deriving the final bound.}
Before observing the transcript, $U_\ell,V_\ell$ are not only independent but also uniform, because $U,V$ are independent and uniform and $E_\ell$ is
surjective. Each pair $(\hat u,\hat v)$ thus has probability $2^{-2b}$. We now bound the expectation of $ \Delta\bigl(Q^{\hat u,\hat v}_{\mathbf x,\mathbf y},Q_{\mathbf x,\mathbf y}\bigr)$.
In the calculation below, all probabilities are conditioned on the
fixed $(\mathbf x,\mathbf y)$:
\[
\begin{aligned}
 \mathbb E_{\hat u,\hat v\leftarrow\mathsf U_b}
 \Delta\bigl(Q^{\hat u,\hat v}_{\mathbf x,\mathbf y},Q_{\mathbf x,\mathbf y}\bigr)
 &=\frac12\sum_{\hat u,\hat v}2^{-2b}\sum_\pi
 \left|\Pr[\Pi=\pi\mid U_\ell=\hat u,V_\ell=\hat v]
              -\Pr[\Pi=\pi]\right|\\
 &=\frac12\sum_{\pi,\hat u,\hat v}
 \left|\Pr[\Pi=\pi,U_\ell=\hat u,V_\ell=\hat v]
              -\Pr[\Pi=\pi]\,2^{-2b}\right|\\
 &=\frac12\sum_\pi\Pr[\Pi=\pi]\sum_{\hat u,\hat v}
 \left|\Pr[U_\ell=\hat u,V_\ell=\hat v\mid\Pi=\pi]
              -2^{-2b}\right|\\
 &=\mathbb E_{\pi\leftarrow\Pi^{\mathbf x,\mathbf y,\ell}}
 \Delta\bigl((U_\ell,V_\ell)^{\mathbf x,\mathbf y,\ell,\pi},
             \mathsf U_b\otimes\mathsf U_b\bigr)\\
 &<5\cdot10^{-11}.
\end{aligned}
\]
Finally, Markov's inequality with $\tau=6\cdot10^{-5}$ gives
\[
 \Pr_{\hat u,\hat v\leftarrow\mathsf U_b}
 \left[\Delta\bigl(Q^{\hat u,\hat v}_{\mathbf x,\mathbf y},
                   Q_{\mathbf x,\mathbf y}\bigr)>\tau\right]
 <\frac{5\cdot10^{-11}}{6\cdot10^{-5}}<10^{-6}.
\]
This proves~\eqref{lb:fiber-tail}.

\section{Applications}
In this section we apply Theorem \ref{thm:main} to prove Theorems \ref{th:shape} and \ref{th:forrelation}. Theorem \ref{th:shape} is obtained by using the $\Shape_m$ function from \cite{Gavinsky2}, as in \cite{gavinsky}.
Theorem~\ref{th:forrelation} uses Bansal--Sinha's $k$-Forrelation
function~\cite{bansal-sinha} composed with an inner product gadget.

\subsection{Application 1: total function from \texorpdfstring{$\boldsymbol{\Shape}$}{Shape}}\label{sec:shape}

We first define the function $\Shape_m$ from~\cite{Gavinsky2}. 
Write Alice's input as $x=(x_1,x_2)$ and Bob's as $y=(y_1,y_2)$, with all four strings of length $m\geq2$. Index positions by $\{0,1,\ldots,m-1\}$ and use the convention $(\sigma_i(a))_k=a_{k-i\bmod m}$. Define
\[
 a_1=x_1\oplus y_1,\qquad a_2=x_2\oplus y_2,\qquad
 d_i=|\sigma_i(a_1)\oplus a_2|.
\]
The function $\Shape_m$ is defined on inputs for which
some cyclic shift of $a_1$ is at Hamming distance at most $2m/5$ from $a_2$,
or every cyclic shift is at distance between $7m/15$ and $8m/15$ from $a_2$.
It outputs $1$ in the first case and $0$ in the second, and is undefined
otherwise. 

The domain predicate of $\Shape_m$ is thus
\begin{equation}\label{eq:shape-domain}
 D_m(x,y)=
 \left(\bigvee_{i=0}^{m-1}[5d_i\leq2m]\right)
 \vee
 \left(\bigwedge_{i=0}^{m-1}[7m\leq15d_i\leq8m]\right).
\end{equation}
The integer comparisons in~\eqref{eq:shape-domain} express the original thresholds $2m/5$, $7m/15$, and $8m/15$ exactly, including when $m$ is not divisible by $15$.

\begin{proposition}\label{prop:shape-circuit}
The domain predicate $D_m$ has a Boolean circuit of size
$O(m\log^2 m)$.
\end{proposition}

\begin{proof}
We show that all $m$ distances $d_i$ can be computed exactly in $O(m\log^2 m)$ bit operations and by an explicit Boolean circuit of size $O(m\log^2 m)$. Consequently, the domain predicate $D_m$ has a circuit of this size.

Define the cyclic overlap counts
\[
 \rho_i=\sum_{j=0}^{m-1}(a_1)_j(a_2)_{(j+i)\bmod m}.
\]
The elementary identity $u\oplus v=u+v-2uv$ for bits gives
\begin{equation}\label{eq:distance-correlation}
 d_i=|a_1|+|a_2|-2\rho_i.
\end{equation}
Thus the overlap counts $\rho_i$ determine every $d_i$. We compute
all the counts $\rho_i$ through an exact integer product, as follows.

Let $w=\lceil\log_2(m+1)\rceil$ and $B=2^w>m$, and pack the two strings into integers
\[
 A=\sum_{j=0}^{m-1}(a_1)_jB^{m-1-j},\qquad
 C=\sum_{k=0}^{m-1}(a_2)_kB^k.
\]
Their product has the expansion
\begin{equation}\label{eq:packed-product}
 AC=\sum_{r=0}^{2m-2}\gamma_rB^r,\qquad
 \gamma_r=\sum_{\substack{0\leq j,k<m\\m-1-j+k=r}}(a_1)_j(a_2)_k.
\end{equation}
Each coefficient satisfies $0\leq\gamma_r\leq m<B$. There are therefore no carries between the coefficient blocks: $\gamma_r$ is precisely the corresponding $w$-bit block of the binary product. Separating the pairs with $k-j=i$ and $k-j=i-m$ gives
\begin{equation}\label{eq:wrap-correlation}
 \rho_i=\gamma_{m-1+i}+\gamma_{i-1}\qquad(0\leq i<m),
\end{equation}
where $\gamma_{-1}=0$. 

The two packed integers have at most $L=mw=O(m\log m)$ bits. By using Harvey and van der Hoeven's integer-multiplication algorithm~\cite{hvdh}, we can compute their product with a Boolean circuit of size
\[
 O(L\log L)=O(m\log^2 m).
\]
Packing and extracting fixed blocks are wiring operations in the circuit. Computing the two weights, the $m$ sums in~\eqref{eq:wrap-correlation}, the distances in~\eqref{eq:distance-correlation}, and the threshold comparisons in~\eqref{eq:shape-domain} costs an additional $O(m\log m)$ Boolean gates. The same postprocessing can be carried out with sequential scans of the coefficient blocks in $O(m\log m)$ bit operations.
\end{proof}


\begin{proof}[Proof of Theorem~\ref{th:shape}]
Gavinsky~\cite{Gavinsky2} established the bounds
\[
 \R(\Shape_m)=\Omega(\sqrt m),\qquad
 \Q^1(\Shape_m)=O(\log^2m).
\]
Apply Theorem~\ref{thm:main} with $n=2m$ and the exact domain-circuit
bound $s=O(m\log^2 m)$ from Proposition~\ref{prop:shape-circuit}.
Then $N=\widetilde \Theta (m)$, the randomized lower bound is
$\widetilde\Omega(\sqrt m)=\widetilde\Omega(\sqrt N)$, and the
quantum upper bound is polylogarithmic in $N$ with two messages.
\end{proof}

\subsection{Application 2: total functions from \texorpdfstring{$\boldsymbol{k}$}{k}-Forrelation}
\label{sec:forrelation}

Fix a constant integer $k\geq2$, let $m\geq2$ be a power of two, and put
$\delta_k=2^{-5k}$. Write $W_m$ for the unnormalized Walsh--Hadamard
matrix, with entries $(W_m)_{a,b}=(-1)^{a\cdot b}$ for
$a,b\in\bits^{\log m}$, and put $H_m=W_m/\sqrt m$.
For $z=(z_1,\ldots,z_k)\in\{\pm1\}^{km}$, define
\[
 \Phi_{k,m}(z)=\frac{1}{m^{(k+1)/2}}
 \mathbf1^\top D_{z_1}W_mD_{z_2}W_m\cdots W_mD_{z_k}\mathbf1,
 \qquad D_{z_i}=\operatorname{diag}(z_i).
\]
The partial function $\mathsf{Forr}_{k,m}$ is $1$ when
$\Phi_{k,m}\geq\delta_k$, and $0$ when
$|\Phi_{k,m}|\leq\delta_k/2$; it is undefined otherwise.
These are the thresholds in Bansal--Sinha's lower bound~\cite{bansal-sinha}.

We compose this function with inner product on
$d=\Theta(\log(km))$ bits, with a sufficiently large constant as
in~\cite[Section~1.4, equation~(1.5)]{ssw}.
Alice and Bob each receive $km$ strings $x_{i,j},y_{i,j}\in\bits^{d}$.
Set
\[
 z_i(j)=(-1)^{\langle x_{i,j},y_{i,j}\rangle_{\F_2}},
 \qquad f_{k,m}(x,y)=\mathsf{Forr}_{k,m}(z).
\]
The input length per party is $kmd=\widetilde O_k(m)$.

The query lower bound of Bansal--Sinha~\cite[Corollary~1.4]{bansal-sinha}
and IP lifting~\cite[Theorem~1]{cfkmp}
(see also~\cite[equation~(1.5)]{ssw}) give
\begin{equation}\label{eq:forrelation-lifted-lower}
 \R(f_{k,m})=\widetilde\Omega_k(m^{1-1/k}).
\end{equation}
The quantum algorithm in~\cite[Fact~5.1]{ssw}, with $U=H_m$, and the
standard query-to-communication simulation~\cite{bcw}
(see also~\cite[equation~(1.4)]{ssw}) give
\begin{equation}\label{eq:forrelation-lifted-upper}
 \Q^{\,2\lceil k/2\rceil}(f_{k,m})=O_k(\log m).
\end{equation}
For the message count, run $O_k(1)$ independent copies in parallel and
threshold their outcomes: the acceptance probabilities in the two promise
cases differ by at least $\delta_k/4$.
Each of the $\lceil k/2\rceil$ query layers requires two messages of
$O(d+\log(km))$ qubits per copy, with no prior entanglement.

\begin{proposition}\label{prop:forrelation-domain}
The domain indicator of $f_{k,m}$ has a Boolean circuit of size
\[
 s=O(kmd+k^2m\log^2 m)=\widetilde O_k(m).
\]
\end{proposition}

\begin{proof}
First compute the $km$ inner products using $O(kmd)$ Boolean gates.
For the resulting signs $z_i(j)$, compute the integer
\[
 T=\mathbf1^\top D_{z_1}W_mD_{z_2}W_m\cdots W_mD_{z_k}\mathbf1
\]
from right to left. 
The fast Walsh--Hadamard transform computes the product of $W_m$
with a vector using $O(m\log m)$ integer additions and subtractions.
Thus the entire computation uses
$O(km\log m)$ integer additions, subtractions, and conditional
sign changes.

Every intermediate integer has absolute value at most $m^k$:
each transform can increase the largest coordinate magnitude by a factor of at most
$m$, and the final sum contributes at most one more factor $m$.
Consequently $O(k\log m)$ signed bits suffice throughout.
Ripple-carry addition and subtraction, as well as conditional negation,
have Boolean circuits of size linear in this bit length. Implementing these operations as Boolean circuits therefore
computes $T$ with $O(k^2m\log^2 m)$ gates.

Since $T$ is an integer and $\delta_k=2^{-5k}$, the two promise branches
are equivalent to the exact integer tests
\begin{align*}
 \Phi_{k,m}\geq\delta_k
 &\quad\Longleftrightarrow\quad
 T\geq0\ \text{ and }\ 2^{10k}T^2\geq m^{k+1},\\
 |\Phi_{k,m}|\leq\delta_k/2
 &\quad\Longleftrightarrow\quad
 2^{10k+2}T^2\leq m^{k+1}.
\end{align*}
Their disjunction recognizes $\dom(f_{k,m})$.
Schoolbook squaring uses $O(k^2\log^2 m)$ gates; shifts by the displayed
powers of two and comparisons with the hardwired integer $m^{k+1}$
use no larger order of resources. Adding the $O(kmd)$ gates needed to compute the
inner products, the total circuit size is
$O(kmd+k^2m\log^2 m)$, as claimed. 
\end{proof}

\begin{proof}[Proof of Theorem~\ref{th:forrelation}]
Apply Theorem~\ref{thm:main} to $f_{k,m}$. Its input length is
$\widetilde O_k(m)$, and Proposition~\ref{prop:forrelation-domain}
gives $s=\widetilde O_k(m)$. Together with
\eqref{eq:forrelation-lifted-lower}--\eqref{eq:forrelation-lifted-upper},
the conversion gives a total function with $N=\widetilde \Theta_k(m)$
input bits per party, randomized complexity
$\widetilde\Omega_k(m^{1-1/k})$, and polylogarithmic quantum complexity
using $2\lceil k/2\rceil+1$ messages. 
\end{proof}

\section*{Acknowledgments}
The author is grateful to Dmytro Gavinsky for sharing his manuscript and for helpful correspondence.
The author is supported by JSPS KAKENHI grants JP24H00071 and JP25K24674, MEXT Q-LEAP grant JPMXS0120319794, JST ASPIRE grant JPMJAP2302, and JST CREST grant JPMJCR24I4. 

\section*{AI disclosure}
GPT-6 Astra was used in the development of this paper,
including in exploring and refining proof ideas and in drafting and revising
the text. In particular, GPT-6 Astra proposed the core proof idea underlying
Theorem \ref{thm:main} in response to the author's prompts. The author subsequently
verified, simplified, and organized the argument and edited its exposition.
The author takes full responsibility for the correctness and content of the paper.

\end{document}